\documentclass{amsart}
\usepackage{amsmath}
\usepackage{paralist}
\usepackage[misc]{ifsym}
\usepackage{epsfig}
\usepackage{epstopdf}
\usepackage[colorlinks=true]{hyperref}
\hypersetup{urlcolor=blue, citecolor=red}
\allowdisplaybreaks

\newtheorem{theorem}{Theorem}[section]

\newtheorem{lemma}[theorem]{Lemma}
\newtheorem{proposition}[theorem]{Proposition}

\newtheorem*{problem}{Problem}
\theoremstyle{definition}
\newtheorem{definition}[theorem]{Definition}
\newtheorem{remark}[theorem]{Remark}
\newtheorem{assumption}[theorem]{Assumption}

\newtheorem{example}[theorem]{Example}
\usepackage{amsfonts,amssymb,mathtools}
\usepackage{colortbl}
\usepackage{blkarray}
\usepackage{makecell}
\usepackage{booktabs}
\usepackage{mathabx}
\usepackage{mathrsfs}
\usepackage{url}
\usepackage{color}
\usepackage{stmaryrd}
\usepackage{multirow}
 \usepackage{algorithm}
 \usepackage[noend]{algpseudocode}
 \usepackage{tikz}
\usepackage{tikz-cd}
\usetikzlibrary{tikzmark,calc,arrows, shapes, decorations.pathreplacing, positioning, quotes}
\usepackage{breakcites}
\usepackage{thmtools}
\usepackage{thm-restate}
\usepackage{mathdots}

\renewcommand{\leq}{\leqslant}

\newcommand{\K}{\ensuremath{\mathbb{K}}}
\newcommand{\F}{\ensuremath{\mathbb{F}}}
\newcommand{\Fq}{\ensuremath{\mathbb{F}_q}}

\newcommand{\Fqm}{\ensuremath{\mathbb{F}_{q^m}}}

\newcommand{\mat}[1]{\ensuremath{\boldsymbol{#1}}}
\newcommand{\code}[1]{\ensuremath{\mathscr{#1}}}

\newcommand{\AC}{\code{A}}

\newcommand{\CC}{\code{C}}
\newcommand{\DC}{\code{D}}

\newcommand{\VC}{\code{V}}

\newcommand{\Am}{\mat{A}}
\newcommand{\Bm}{\mat{B}}

\newcommand{\Gm}{\mat{G}}
\newcommand{\Hm}{\mat{H}}
\renewcommand{\Im}{\mat{I}}
\newcommand{\Mm}{\mat{M}} 
        
\newcommand{\Pm}{\mat{P}}

\newcommand{\Sm}{\mat{S}}
\newcommand{\Vm}{\mat{V}}

\newcommand{\zerov}{\mat{0}}
\newcommand{\av}{\mat{a}}
\newcommand{\bv}{\mat{b}}
\newcommand{\cv}{\mat{c}}
\newcommand{\dv}{\mat{d}}

\newcommand{\mv}{\mat{m}}

\newcommand{\vv}{\mat{v}}
\newcommand{\wv}{\mat{w}}
\newcommand{\xv}{{\mat{x}}}
\newcommand{\yv}{{\mat{y}}}
\newcommand{\zv}{{\mat{z}}}

\newcommand{\cA}{\mathcal{A}}
\newcommand{\cB}{\mathcal{B}}

\newcommand{\cI}{\mathcal{I}}
\newcommand{\cP}{\mathcal{P}}

\newcommand{\cV}{\mathcal{V}}

\newcommand{\GRS}[3]{\text{\bf GRS}_{#1}(#2,#3)}
\newcommand{\Alt}[3]{\code{A}_{#1}(#2, #3)}
\newcommand{\Goppa}[2]{\code{G}(#1, #2)}

\newcommand{\trsp}[1]{{{#1}^{\intercal}}}

\newcommand{\Crel}{\code{C}_\text{rel}}
\newcommand{\Cmat}{\code{C}_\text{mat}}

\newcommand{\starp}[2]{{#1} \star {#2}}

\newcommand{\sq}[1]{#1^{\star 2}}
\newcommand{\sqb}[1]{\left(#1\right)^{\star 2}}

\newcommand{\GL}{\mathbf{GL}}

\newcommand{\rank}{\mathbf{Rank}}

\newcommand{\eqdef}{\stackrel{\text{def}}{=}}
\newcommand{\ie}{i.e. }
\newcommand{\Span}[2]{\left\langle \, #1 \, \right\rangle_{#2}}
\newcommand{\Fspan}[1]{\left\langle \, #1 \, \right\rangle_{\F}}

\newcommand{\floor}[1]{\left\lfloor #1 \right\rfloor}
\newcommand{\ceil}[1]{\left\lceil #1 \right\rceil}

\newcommand{\Iintv}[2]{\llbracket #1 , #2 \rrbracket}
\newcommand{\card}[1]{\lvert #1 \rvert}

\newcommand{\pf}{pf}

\newcommand{\Sym}{\mathbf{Sym}}
\newcommand{\Skew}{\mathbf{Skew}}

\newcommand{\ext}[2]{#1_{#2}}

\title[On the matrix code of quadratic relationships for a Goppa code]
{On the matrix code of quadratic relationships for a Goppa code} 

\author[Rocco Mora]{}

\begin{document}
\maketitle

\centerline{\scshape
Rocco Mora}

\medskip

{\footnotesize
 \centerline{CISPA – Helmholtz Center for Information Security, Germany}
} 

\bigskip


\begin{abstract}
In this article, we continue the analysis started in \cite{CMT23} for the matrix code of quadratic relationships associated with a Goppa code. We provide new sparse and low-rank elements in the matrix code and categorize them according to their shape.
Thanks to this description, we prove that the set of rank 2 matrices in the matrix codes associated with square-free binary Goppa codes, \ie those used in the Classic McEiece cryptosystem, is much larger than what is expected, at least in the case where the Goppa polynomial degree is 2.
We build upon the algebraic determinantal modeling introduced in \cite{CMT23} to derive a structural attack on these instances.
Our method can break in just a few seconds some recent challenges to recover the key of the McEliece cryptosystem, consistently reducing their estimated security level.
We also provide a general method, valid for any Goppa polynomial degree, to transform a generic pair of support and multiplier into a pair of support and Goppa polynomial.
\end{abstract}


\section{Introduction}

\subsection*{The McEliece scheme and its cryptanalysis}

The McEliece cryptosystem \cite{M78} is the oldest code-based encryption scheme, dating back to 1978, \ie just a few months after the ubiquitously used RSA cryptosystem \cite{RSA78}. Contrarily to the latter \cite{S94a}, the former is also widely believed to be a quantum-resistant alternative, meaning that quantum algorithms are not expected to break it exponentially faster than classical ones. This is mirrored in the NIST Post-Quantum Standardization Process, where the IND-CCA secure version Classic McEliece \cite{ABCCGLMMMNPPPSSSTW20} is currently one of the few candidates in the fourth round. Despite the public key size being huge, the McEliece encryption scheme benefits from extremely fast encryption and decryption algorithms and very small ciphertexts. This potentially makes it an attractive option for several use-cases\footnote{\url{https://groups.google.com/a/list.nist.gov/g/pqc-forum/c/jgevyeKehcM}}. 

The other important argument in favor of McEliece is that all the known general decoding algorithms for message recovery developed in over 60 years of research, both classical or quantum, barely improved the exponent of the exponential cost \cite{P62,S88,D89,CC98,MMT11,BJMM12,MO15,BM17,B10,KT17a}. However, these are still used to design secure parameters, because key recovery attacks are immensely more expensive than message recovery techniques. 

Key recovery attacks try to exploit the algebraic structure of the underlying family of codes. Indeed, in order to decrypt a message, the receiver must be able to decode a codeword and therefore a code equipped with an efficient decoding algorithm must be adopted. The original proposal of McEliece, as well as Classic McEliece, builds upon the class of binary Goppa codes. An element of this family is uniquely determined by a vector, called support, of length equal to the code length and an irreducible polynomial of a relatively small degree defined over an extension of the binary field which is called Goppa polynomial. On the other hand, a Goppa code corresponds to several pairs of supports and Goppa polynomials. Recovering any of them allows to decode efficiently any message. 

For a long time, the only key recovery attack consisted of guessing a valid pair of support and Goppa polynomial \cite{LS01} and then checking via the Support Splitting Algorithm \cite{S00} whether it defines a Goppa code that is permutation equivalent to the public one. The total complexity is exponential and, as already mentioned, the exponent is much bigger than the one for message recovery approaches. 

Even the potentially easier task of distinguishing efficiently if a generator matrix comes from a Goppa code or a random one (this takes the name of the Goppa distinguishing problem) had been considered difficult for a long time. This is because Goppa codes share a lot of properties in common with random ones. For instance, they asymptotically meet the Gilbert-Varshamov bound, they have approximately the same weight distribution and also a trivial permutation group. Based on the two assumptions: (1) the pseudorandomness of Goppa codes and (2) the hardness of decoding up to the binary Goppa bound, it is possible to devise a reductional proof of security for the McEliece scheme \cite{S10}. 

An important step in the understanding of the structure of a Goppa code is the high-rate distinguisher presented in \cite{FGOPT11}. Here it was shown that a linear system associated with Goppa codes has an unusually small rank when the code rate is high enough. This is not the case for Classic McEliece, but it impacts other schemes like the CFS digital signature \cite{CFS01}. Much later, a different perspective about the distinguisher was given in \cite{MT23}, exploiting the link between the linear system and square codes revealed in \cite{MP12}. More precisely, the distinguisher was explained in terms of the dimension of the \textit{square code} of the dual of a Goppa code. Tight upper bounds for this dimension have been provided in \cite{MT23}, thus making the distinguisher more rigorous.
The square code analysis given in \cite{MT23} resulted in the first-ever polynomial-time cryptanalysis on unstructured alternant codes with high rate \cite{BMT23}. 
This approach is limited to the high-rate regime and the problem of attacking, or even simply distinguishing, a Goppa code with rate comparable to those used in Classic McEliece was left open. This question has been partially addressed in \cite{CMT23}, where a completely new approach based on quadratic forms to attack the McEliece cryptosystems has been introduced.

\subsection*{Contributions of this work and organization of the paper}

The contribution of the paper is twofold. In a nutshell:
\begin{enumerate}
	\item We illustrate a new efficient key-recovery attack on binary Goppa codes with a Goppa polynomial degree equal to 2. The attack can be split into two parts. First, we perform an algebraic cryptanalysis on the Pfaffian system introduced in \cite{CMT23}, thus finding low-rank elements in the matrix code of quadratic relationships. Then we exploit the knowledge of such matrices to reconstruct the secret key of the Goppa code, i.e. a valid pair of support and Goppa polynomial. As we will explain, this attack is tailored specifically for Goppa codes and does not affect generic alternant codes of order 2. 
	\item The previous result is made possible by an in-depth analysis of structured elements lying within the matrix code of quadratic relationships originated by an alternant or a Goppa code. Indeed, we prove that, when the Goppa polynomial degree is 2, the variety associated with the Pfaffian system is big. This fact is instrumental for the algebraic attack as it shows that several variables can be specialized. Our investigation, however, is more general than the attack, as it covers any alternant or Goppa code degree.
\end{enumerate}

The focus on binary Goppa codes with a Goppa polynomial degree equal to 2 of the first contribution has been partially motivated by some recent key-recovery challenges for the McEliece scheme\footnote{\url{https://www.herox.com/TIIMcElieceChallenges}, track 1B.}, that we will call ``TII challenges'' from now on. We have been able to break all TII challenges using a Goppa polynomial of degree $r=2$ and length $n> 3rm-3$, where $m$ is the field extension degree, within a few seconds, as it can be consulted at \url{https://www.herox.com/TIIMcElieceChallenges/leaderboard}. Notably, one of them has claimed bit complexity $\lambda=68$. 

After recalling some preliminary notions and results in Section 2, we introduce in Section 3 the ``Goppa code representations'' of a Goppa code, which is useful for recovering the Goppa polynomial and corresponding support. In Section 4, we categorize new matrices in the matrix code of quadratic relations and use them to study the variety associated with the Pfaffian ideal in the special case of $r=2$. This analysis led us in Section 5 to mount a polynomial-time attack on square-free binary Goppa codes of degree 2. 

\section{Preliminaries}

\subsection{Notation}

\subsubsection*{General notation.}
The closed integer interval between $a$ and $b$ is denoted with $\Iintv{a}{b}$.

\subsubsection*{Finite fields.}
We denote by $\K$ a generic field and by $\overline{\K}$ its algebraic closure. Instead, $\F$ stands for a generic finite field and $\Fq$ for the finite field of size $q$, where $q$ is a prime power. We will often consider the finite field extension $\Fqm/\Fq$, where $\Fqm$ is the finite fields with $q^m$ elements, for some positive integer $m$. We denote by $\F^*$ the set of nonzero elements of the finite field $\F$, which form a multiplicative group.

\subsubsection*{Vectors and matrices.}
Vectors are indicated by lowercase bold letters $\xv$ and matrices by uppercase bold letters $\Mm$. By convention, vector coordinates are indexed starting from 1, but we will write matrix blocks using indices that start from 0. We denote the component-wise image with respect to a vector $\xv=(x_i)_{1\le i \le n} \in \F$ of a function $f$ with domain $\F$ by using the expression $f(\xv)$, i.e. $f(\xv)=(f(x_i))_{1\le i \le n}$. In a similar manner, given $\xv,\yv\in\F^n$ and two positive integers $a,b$, we denote by $\xv^a\yv^b$ the vector $(x_i^a y_i^b)_{1 \leq i \leq n}$. Given a matrix $\Mm=(m_{i,j})\in \Fqm^{m\times n}$, we write $\Mm^{(q)}=(m_{i,j}^q)$, \ie the matrix where the Frobenius automorphism $a\mapsto a^q$ has been applied to all the entries. The set of $k \times k$ symmetric matrices over $\F$ is denoted by $\Sym(k,\F)$, whereas the corresponding set of skew-symmetric matrices is denoted by $\Skew(k,\F)$. 

\subsubsection*{Vector spaces.}
The $\K$-linear space generated by the not necessarily linearly independent vectors $\xv_1,\dots,\xv_m \in \K^n$ is denoted by $\Span{\xv_1,\dots,\xv_m}{\K}$. If $\K=\F$ then $\CC=\Span{\xv_1,\dots,\xv_m}{\F}$ is an $[n,k]$-linear code, where $k$ denotes the dimension of $\CC$.

\subsubsection*{Polynomial ideals.}
Polynomial ideals are indicated by calligraphic capital letters. Given the multivariate polynomials $f_1,\dots,f_m\in\K[x_1,\dots,x_n]$, we denote by $\langle f_1,\dots,f_m\rangle$  the polynomial ideal generated by them. The variety associated with a polynomial ideal $\cI\subseteq \K[x_1,\dots,x_n]$ is $\Vm(\cI)=\{\av \in \overline{\K}^n \mid \forall f \in \cI,\; f(\av)=0\}$.

\subsection{GRS and Goppa codes}
We first recall the definition of GRS codes, a family of evaluation codes.
\begin{definition}[Generalized Reed-Solomon (GRS) code ]\label{def:GRS}
	Let $\xv=(x_1,\dots,x_n)\in\F^n$ be a vector of pairwise distinct entries and $\yv=(y_1,\dots,y_n)\in\F^n$ a vector of nonzero entries. The \textit{generalized Reed-Solomon (GRS) code} over $\F$ of dimension $k$ with \textit{support} $\xv$ and \textit{multiplier} $\yv$ is
	\[
	\GRS{k}{\xv}{\yv}\eqdef\{(y_1 P(x_1),\dots,y_n P(x_n)) \mid P \in \F[z], \deg P < k\}.
	\]
\end{definition}
An alternant code is defined as the subfield subcode of a GRS code. Here we exploit the following proposition to define the former as the subfield subcode of the dual of a GRS code.
\begin{proposition} \cite[Theorem~4, p.~304]{MS86}\label{pr:dual_GRS} 
	Let $\GRS{r}{\xv}{\yv}$ be a GRS code of length $n$. Its dual is also a GRS code. In particular
	$
	\GRS{r}{\xv}{\yv}^\perp=\GRS{n-r}{\xv}{\yv^\perp},
	$
	with $
	\yv^\perp\eqdef\left(\frac{1}{\pi'_\xv(x_1)y_1},\dots,\frac{1}{\pi'_\xv(x_n)y_n}\right)$, where $\pi_\xv(z)\eqdef \prod_{i=1}^n (z-x_i)$ and $\pi'_\xv$ is its derivative.
\end{proposition}

\begin{definition}[Alternant code]
	Let $n\le q^m$, for some positive integer $m$. Let $\GRS{r}{\xv}{\yv}$ be the GRS code over $\Fqm$ of dimension $r$ with support $\xv \in \Fqm^n$ and multiplier $\yv\in (\Fqm^*)^n$. The \textit{alternant code} with support $\xv$, multiplier $\yv$ and \textit{degree} $r$ over $\Fq$ is
	\[
	\Alt{r}{\xv}{\yv}\eqdef \GRS{r}{\xv}{\yv}^\perp \cap \F_q^n=\GRS{n-r}{\xv}{\yv^\perp} \cap \Fq^n.
	\]
	The integer $m$ is called \textit{extension degree} of the alternant code.
\end{definition}
A Goppa code is an alternant code where the support and multiplier are linked by a very particular relation.
\begin{definition}[Goppa code]
	Let $\xv\in\Fqm^n$ be a support vector and $\Gamma\in\Fqm[z]$ a polynomial of degree $r$ such that $\Gamma(x_i)\neq 0$ for all $i \in \{1,\dots,n\}$. The \textit{Goppa code} of degree $r$ with support $\xv$ and \textit{Goppa polynomial} $\Gamma$ is defined as
	$
	\Goppa{\xv}{\Gamma}\eqdef\Alt{r}{\xv}{\yv},$
	where $\yv\eqdef\left(\frac{1}{\Gamma(x_1)},\dots,\frac{1}{\Gamma(x_n)}\right).$
\end{definition}

\begin{theorem}\label{thm: binary_Goppa->Alt} \cite{P75}
	Let $\Goppa{\xv}{\Gamma}$ be a binary Goppa code with a square-free Goppa polynomial $\Gamma$ of degree $r$. Then
	\[\Goppa{\xv}{\Gamma}=\Goppa{\xv}{\Gamma^2}=\Alt{2r}{\xv}{\yv},\]
	where $y_i\eqdef\frac{1}{\Gamma(x_i)^2}$ for all $1\le i \le n$.
\end{theorem}

\subsection{Product and square of codes}
The notion of squares of codes is at the core of the high-rate distinguisher as presented in \cite{MT23}.
Given the \textit{component-wise product} of two vectors $\av,\bv\in\F^n$
\[
\starp{\av}{\bv}\eqdef(a_1 b_1,\dots,a_n b_n),
\]
we define the component-wise (or Schur's) product of codes.
\begin{definition}
	The  \textit{component-wise product of codes} $\CC,\DC$ over $\F$ with the same length $n$ is defined as
	\[
	\starp{\CC}{\DC}\eqdef \Span{\starp{\cv}{\dv} \mid \cv \in \CC, \dv \in \DC}{\F}.       \]
	If $\CC=\DC$, we call $\sq{\CC}\eqdef\starp{\CC}{\CC}$ the \textit{square code} of $\CC$. 
\end{definition}

\subsection{Extension of a code over a field extension}
The Frobenius map can be applied component-wise over codes in the following way.
\begin{definition}[Image of a code by the Frobenius map]
	Let $\CC\in \Fqm^n$ be a code and $i$ a non-negative integer. We define by $\CC^{(q^i)}$ as
	\[
	 \CC^{(q^i)}\eqdef \{ (c_1^{q^i},\dots,c_n^{q^i}) \mid (c_1,\dots,c_n)\in \CC\}.
	\]
\end{definition}
For some codes naturally defined over $\Fq$, namely subfield subcodes, we will extensively consider their linear span over a field extension $\Fqm$. More formally, we define the extension of a code over a field extension (or extension of scalars) in the following way.
\begin{definition}[Extension of a code over a field extension]
	Let $\CC$ be a linear code over $\Fq$. We denote by $\CC_{\Fqm}$ the $\Fqm$-linear span of $\CC$ in $\Fqm^n$.
\end{definition}
If we apply this construction on the dual of an alternant code, we get 
\begin{proposition} \label{prop:dual_alt_fqm} \cite{BMT23}
	Let $\Alt{r}{\xv}{\yv}$ be an alternant code over $\Fq$. Then \\
	$
	\left(\Alt{r}{\xv}{\yv}^\perp\right)_{\Fqm} =\sum_{j=0}^{m-1} \GRS{r}{\xv}{\yv}^{(q^j)}= \sum_{j=0}^{m-1} \GRS{r}{\xv^{q^j}}{\yv^{q^j}}.$
\end{proposition}
This is useful because it allows to view the code generators as the component-wise evaluations of monomials, as this is the case for the GRS codes $ \GRS{r}{\xv^{q^j}}{\yv^{q^j}}$. This perspective is motivated and made possible by the fact that the extension of scalars commutes with all the standard constructions in coding theory. We recall from \cite{R15} the properties that will be implicitly exploited in this work.
\begin{lemma}[from Lemma 2.22 and Lemma 2.23, \cite{R15}]
	Let $\CC,\CC'\subseteq \Fq^n$ be two $\Fq$-linear codes. Then
	\begin{itemize}
		\item If $\Gm$ is a generator matrix for $\CC$ over $\Fq$, then it is also a generator matrix of $\CC_{\Fqm}$ over $\Fqm$.
		\item If $\Hm$ is a parity-check matrix for $\CC$ over $\Fq$, then it is also a parity-check matrix of $\CC_{\Fqm}$ over $\Fqm$.
		\item $(\CC^\perp)_{\Fqm}=(\CC_{\Fqm})^\perp \subseteq \Fqm^n$.
		\item $\CC\subseteq \CC' \iff \CC_{\Fqm}\subseteq \CC'_{\Fqm}$.
		\item $(\CC+\CC')_{\Fqm}=\CC_{\Fqm}+\CC'_{\Fqm}$.
		\item $(\CC\cap\CC')_{\Fqm}=\CC_{\Fqm}\cap\CC'_{\Fqm}$.
		\item $(\starp{\CC}{\CC'})_{\Fqm}=\starp{\CC_{\Fqm}}{\CC'_{\Fqm}}$.
	\end{itemize}
\end{lemma}

\subsection{The matrix code of quadratic relationships}
In \cite{CMT23}, a new object is associated with a linear code: the matrix code of quadratic relationships. This notion is strictly related to that of the square code. If $\CC=\Fspan{\cv_1,\dots,\cv_k}\subseteq \F^n$ is a $k$-dimensional linear code, then a natural set of generators for its square code $\sq{\CC}$ is given by all component-wise products of codewords $\starp{\cv_i}{\cv_j}$, $1\le i\le j \le k$. Some of these generators may be linearly dependent and this is even expected for random codes when the dimension $k$ is big enough compared to $n$, so that the whole space is reached by the square code. When $\CC$ is the dual of an alternant or Goppa code, some very structured linear dependencies appear among the generators of $\sq{\CC}$. In other words, quadratic relationships among the generators of $\CC$ are guaranteed to exist. The high-rate distinguisher relies on this fact, by counting the dimension of the space generated by such quadratic dependencies and comparing them with those for a random code, if any. However, unless the rate of the Goppa code is very high, the dimensions of the two corresponding spaces are the same, even if the hidden structure is different. The matrix code serves as a tool to distinguish Goppa codes even when the two dimensions match. 

 In particular, the contributions of \cite{CMT23} can be summarized as follows:
\begin{enumerate}
	\item A general and simple procedure to distinguish alternant and Goppa codes of rate at least 2/3 from random ones has been introduced. This consists of defining a determinantal ideal, \ie an ideal generated by minors of a given size. In the mentioned rate regime, the matrix code of quadratic relations constructed from a random code is not expected to contain nonzero matrices of rank at most 3. Therefore, the variety associated with the ideal of minors of size 4 is trivial. The inverse happens for the matrix code associated with an alternant or a Goppa code: low-rank matrices are guaranteed to exist in it. Hence, solving the polynomial system of minors (or of Pfaffians) allows to discriminate between random and alternant or Goppa codes. The complexity of this new distinguisher smoothly interpolates between polynomial (in the regime already distinguishable by \cite{FGOPT11}) and superexponential for constant rates. In other words, it has subexponential complexity for families of alternant or Goppa codes whose dimension $k$ grows between linear and quadratic with respect to the codimension $n-k$.

	\item A new polynomial-time structural attack on alternant and Goppa codes in the rate of parameters distinguishable by \cite{FGOPT11, MT23} has been devised. This contribution extends the results of \cite{BMT23} where generic alternant codes with binary or ternary field size have been attacked. In particular, the algorithm from \cite{CMT23} in concert with \cite{BMT23} breaks all distinguishable alternant codes (for any field size) and all distinguishable Goppa codes whose Goppa polynomials have degree $r<q-1$.
\end{enumerate}

In the following, we will recall the definition of matrix code of quadratic relations and the main results about it. 
Exploiting Proposition~\ref{prop:dual_alt_fqm}, we define an ordered basis, and call it \textbf{canonical basis}, $\cA$ of $\left(\Alt{r}{\xv}{\yv}^\perp\right)_{\Fqm}$:
\begin{equation} \label{eq: basisA}
	\cA\eqdef (\yv,\xv \yv,\dots,\xv^{r-1} \yv,\dots, \yv^{q^{m-1}},(\xv \yv)^{q^{m-1}},\dots,(\xv^{r-1} \yv)^{q^{m-1}}).
\end{equation}
We remark that $\cA$ is an unknown basis from the attacker's point of view, as he/she does not have access a priori to the basis of a single GRS code $\GRS{r}{\xv}{\yv}^{q^j}$ and even less to the monomial basis. In the attack from \cite{CMT23}, $\cA$ was indeed the secret basis to recover. We also denote with $\cB$ another (public) basis of $\left(\Alt{r}{\xv}{\yv}^\perp\right)_{\Fqm}$.

Observe that, for any $0\le a,b,c,d < r$ such that $a+b=c+d$, we have $\xv^a\yv,\xv^b \yv, \xv^c\yv,\xv^d \yv\in \GRS{r}{\xv}{\yv}$ and these codewords satisfy the following quadratic relationship: 
\begin{equation}\label{eq:quadratic_relation}
	\starp{(\xv^a\yv)}{(\xv^b\yv)} = \starp{(\xv^c\yv)}{(\xv^d\yv)}.
\end{equation}
Proposition~\ref{prop:dual_alt_fqm} explains why the existence of these quadratic relationships is preserved when considering subfield subcodes of GRS codes, \ie alternant (including Goppa) codes. Moreover, other structured relations arise from the subfield subcode construction, because codewords from different $\GRS{r}{\xv}{\yv}^{q^j}$'s may be involved. The following definition captures the fact that in general quadratic relationships form a vector space. 
\begin{definition}[Code of quadratic relationships, \cite{CMT23}] \label{def: crel}
	Let $\CC$ be an $[n,k]$ linear code over $\F$ and let $\cV=\{\vv_1,\dots,\vv_k\}$ be a basis of $\CC$. The \textbf{code of quadratic relationships between the Schur's products with respect to $\cV$} is
	\[
	\Crel(\cV)\eqdef\{\cv=(c_{i,j})_{1\le i\le j \le k} \mid \sum_{i \le j} c_{i,j} \starp{\vv_i}{\vv_j}=0\} \subseteq \F^{\binom{k+1}{2}}.
	\] 
\end{definition}
By considering codewords corresponding to the identities \eqref{eq:quadratic_relation}, we can therefore predict the shape of low (Hamming) weight codewords in the code of quadratic relationships $\Crel(\cA)$, where $\cA$ is given in \eqref{eq: basisA}. This fact has potential interest for cryptanalysis, however the same does not happen in general for $\Crel(\cB)$, where $\cB$ is another basis of $\left(\Alt{r}{\xv}{\yv}^\perp\right)_{\Fqm}$. Not only do we not know the shape of low-weight codewords in $\Crel(\cB)$, but these are not even guaranteed to exist.  This suggests that the Hamming distance is not the right metric to look at and that another perspective is required.

Note that any element $\cv=(c_{i,j})_{1\le i\le j \le k}\in\Crel(\cV)$ defines a quadratic form as
$$
Q_{\cv}(x_1,\cdots,x_k) = \sum_{ i\le j }  c_{i,j} x_i x_j.
$$
For instance, the relationship in \eqref{eq:quadratic_relation} is associated with the low-rank quadratic form
\[
x_{a+1} x_{b+1} - x_{c+1} x_{d+1}.
\]
We can therefore represent the elements of $\Crel(\cV)$ as matrices corresponding to the bilinear map given by the polar form of the quadratic form, i.e. the matrix $\Mm_{\cv}$ corresponding to $\cv \in \Crel(\cV)$ that satisfies for 
all $\wv$ and $\zv$ in $\Fqm^k$
\begin{equation}
	\label{eq:polarform}
	\wv \Mm_{\cv} \trsp{\zv} = Q_{\cv}(\wv+\zv) -Q_{\cv}(\wv)-Q_{\cv}(\zv).
\end{equation}
Note that $\Mm_{\cv}$ is symmetric in odd characteristic, whereas it is skew-symmetric in characteristic $2$.
This definition provides a matrix perspective on the space of quadratic relationships that fits in both the odd characteristic and characteristic $2$ cases:
\begin{definition}[Matrix code of relationships, \cite{CMT23}] \label{def: cmat_odd}
	Let $\CC$ be an $[n,k]$ linear code over $\F$ and let $\cV=\{\vv_1,\dots,\vv_k\}$ be a basis of $\CC$. The \textbf{matrix code of relationships between the Schur's products with respect to $\cV$} is
	\[
	\Cmat(\cV)\eqdef\{\Mm_{\cv}=(m_{i,j})_{\substack{1\le i\le k \\ 1\le j\le k}} \mid \cv=(c_{i,j})_{1\le i\le j \le k} \in \Crel(\cV)  \} \subseteq \Sym(k, \F),
	\]
	where  
	$\Mm_{\cv}$ is defined as
	$
	\begin{cases}
		m_{i,j}\eqdef  m_{j,i} \eqdef c_{i,j},&\quad 1\le i< j\le k,\\
		m_{i,i} \eqdef  2c_{i,i},&\quad 1\le i\le k.
	\end{cases}
	$
\end{definition}

The turning point of this approach is that, regardless of the chosen basis, the matrix code of relationships contains low-weight elements when computed with respect to the usual rank metric
\[
d(\Mm_1, \Mm_2)\eqdef \rank(\Mm_1-\Mm_2).
\]
Indeed we have
\begin{proposition}[\cite{CMT23}]  \label{prop: congr_odd}
	Let $\cA$ and $\cB$ be two bases of a same $[n,k]$ $\F$-linear code $\CC$. Then $\Cmat(\cA)$ and $\Cmat(\cB)$ are congruent matrix codes, \ie there exists $\Pm\in \GL_k(\F)$ such that
	\begin{equation}
		\Cmat(\cA)=\trsp{\Pm}\Cmat(\cB) \Pm.
	\end{equation}
	The matrix $\Pm$ coincides with the change of basis matrix between $\cA$ and $\cB$.
\end{proposition}

The proposition above readily implies that the weight distributions of $\Cmat(\cA)$ and $\Cmat(\cB)$ (again with respect to the rank metric) are the same. Another trivial invariant is the matrix code dimension:
\begin{proposition}[\cite{CMT23}] \label{prop:dimension}
	Let $\CC\subseteq \F^n$ be an $[n,k]$ linear code with ordered basis $\cV$. Then
	\[
		\dim_{\F} \Cmat(\cV)=\dim_{\F} \Crel(\cV)=\binom{k+1}{2}-\dim_{\F} \sq{\CC}.
\]
\end{proposition}

If $\cV$ is a basis of $\Alt{r}{\xv}{\yv}^\perp_{\Fqm}$, then $\Cmat(\cV)$ contains rank-3 matrices in odd characteristic and rank 2-matrices in characteristic 2, obtained from \eqref{eq:quadratic_relation} by taking $c=d$. These elements are categorized as Type 1 matrices in the following section.

In \cite{CMT23}, the question of whether matrices of such low rank are expected in $\Cmat(\cV)$, where $\cV$ is the basis of a random $[n,rm]$ $\Fqm$-linear, code has been addressed. 
By computing the Gilbert-Varshamov distance with respect to the space of (skew-)symmetric matrices, it was shown in \cite[Proposition 10]{CMT23} that matrices of rank $\le 3$ belong to $\Cmat(\cV)$ with non-negligible probability iff
\begin{equation} \label{eq: n le 3rm-3}
	n\le 3rm-3.
\end{equation}
Therefore, it is possible to set up an algebraic modeling for the (skew-)symmetric variant of the MinRank problem with target rank 2 or 3.
\begin{problem}[(Skew-)Symmetric MinRank problem for rank $r$]
	Let $\Mm_1,\cdots,\Mm_K$ be (skew-)symmetric matrices in $\F^{N \times N}$. Find an $\Mm \in \Fspan{\Mm_1,\cdots,\Mm_K}$ of rank $r$.
\end{problem}
We remark that:
\begin{itemize}
	\item if one is able to prove that the MinRank instance has no nonzero solution, then it means that the instance is not an alternant (or Goppa) code, which leads to a distinguisher;
	\item all the parameters used in Classic McEliece are such that $n> 3rm-3$.
\end{itemize}
From these two observations, a special MinRank modeling for the case of characteristic 2 has been introduced in \cite{CMT23} and a complexity estimate of the corresponding distinguisher has been determined.

In the rest of the paper, we will deepen the study of the matrix code originated by a generic alternant code or a Goppa code, describing and categorizing structured matrices lying in it. Thanks to this characterization, we also present an adaptation of a polynomial-time attack from \cite{CMT23} to binary Goppa codes of degree 2. 

\section{The Goppa code representation}

Imagine that we are able to recover a pair $(\xv',\yv')$ of valid support and multiplier for the Goppa code $\Goppa{\xv}{\Gamma}$, \ie such that $\Goppa{\xv}{\Gamma}=\Alt{r}{\xv'}{\yv'}$. This task will be the goal of the next sections. In general, $\xv'$ and $\yv'$ do not coincide with $\xv$ and $1/\Gamma(\xv)$ respectively, as the public code is not uniquely determined by a pair of support and multiplier and there is no way to recover the original ones. In addition, $\yv'$ is not even guaranteed to be the inverse of the evaluation over $\xv'$ of a degree-$r$ polynomial $\Gamma'$. The next definition formalizes this concept.
\begin{definition}[Goppa code representation]
	Let $\Goppa{\xv}{\Gamma}\subseteq \Fq^n, \deg(\Gamma)=r$, be an $[n, n-rm]$ Goppa code obtained as the subfield subcode of a Reed-Solomon code defined over an extension field of degree $m$. The pair $(\xv',\yv')$ is said to be an $r$-\textit{Goppa code representation} of $\Goppa{\xv}{\Gamma}$ if $\Goppa{\xv}{\Gamma}=\Alt{r}{\xv'}{\yv'}$ and $\yv'=\frac{1}{\Gamma'(\xv')}$ for some $\Gamma'\in  \Fqm[z]$ of degree $r$.
\end{definition}
\begin{remark}
	The definition above takes into account that Goppa polynomials of different degrees can be defined. For instance any pair of support and multiplier $(\xv',\yv')$ is such $\yv'=\frac{1}{\Gamma'(\xv')}$ for some $\Gamma$ of degree $\le n-1$ because of the interpolation theorem. Moreover, in the binary case, if $\Gamma$ is square-free then Proposition~\ref{thm: binary_Goppa->Alt} provides a new Goppa polynomial of degree $2r$. In other words, the degree $r$ required in the definition is the minimal one, \ie that for which $r=\frac{n-k}{m}$, where $k$ is the Goppa code dimension.
\end{remark}

While any valid pair of support and multiplier permits to decode a Goppa code, there exists a crucial difference that makes Goppa code representations particularly relevant. Again, this distinction is witnessed in the binary square-free Goppa case. In principle, knowing a generic pair $(\xv,\yv)$ enables to decode up to $\frac{r}{2}$ errors. Because of Proposition~\ref{thm: binary_Goppa->Alt}, a Goppa code representation readily allows to see the Goppa code as an alternant code of degree $2r$ and thus leads to an improved decoding capability of $r$ errors. In the McEliece scheme, and more in general in code-based cryptography, the error weight is chosen to be close or equal to the maximum of coordinates that a decoder can correct in such a way to increase the difficulty of decrypting for an attacker. This means that a valid pair $(\xv,\yv)$ for a binary square-free Goppa code is not enough to efficiently decode errors of weight above $r/2$. 

Furthermore, the TII challenges accept solutions in the format `support + Goppa code', thus, translating to our formalism, they implicitly require to find a Goppa code representation.

In the following, we will then show how to move from a generic pair $(\xv,\yv)$ to an equivalent one $(\xv',\yv')$ that is a Goppa code representation. The argument is not limited to $r=2$, but works for any Goppa polynomial degree. This problem has already been partially investigated in \cite[Section 4.4.6]{M23} in relation to the parity-check subcode of a Goppa code, the filtration and Gr\"obner basis computation steps in the distinguisher-based attack \cite{BMT23}. The group of transformations that map a valid pair $(\xv, \yv)$ into another one $(\xv', \yv')$ has been described in \cite{D87} for Cauchy codes and their subfield subcodes. Cauchy codes are the generalization of GRS codes over the projective line $\mathbb{P}^1(\Fqm)$. However, the reconstruction of $(\xv, \yv)$ provided in the previous section guarantees that $\xv\in \Fqm^n\subset \mathbb{P}^1(\Fqm)^n$. We recall from \cite{M23} that, when $\xv\in \Fqm^n$ (and $\yv\in (\Fqm^*)^n)$, the restricted map from \cite{D87} becomes 
\begin{equation} \label{eq: supp_mult_prime} 
	\begin{array}{cccc}
		f \colon & \Fqm^n\times (\Fqm^*)^n& \to & \mathbb{P}^1(\Fqm) \times   (\Fqm^*)^n\\ & (\xv,\yv)& \mapsto &(\xv',\yv')=(\frac{a\xv+b}{c\xv+d}, \lambda (c\xv+d)^{r-1} \yv),
	\end{array} 
\end{equation}
for some $a,b,c,d,\lambda \in \Fqm$, such that $ad-bc\ne 0$ and $\lambda \ne 0$. We want to determine values of $a,b,c,d,\lambda$ that map into admissible vectors $\xv'\in \Fqm^n$ and $\yv'\in (\Fqm^*)^n$. This conditions are equivalent to say that $0\notin \{c x_i +d \mid i \in \Iintv{1}{n}\}$. Moreover, we want $(\xv', \yv')$ to be a Goppa code representation. The property of being a Goppa codes representation is preserved if and only if $f$ is an affine map.
\begin{proposition} \label{prop: f_affine}
	Let $\Goppa{\xv}{\Gamma}=\Alt{r}{\xv}{\yv}$ be an $[n, n-rm]$ Goppa code with $\yv=\frac{1}{\Gamma(\xv)}$ and $\deg(\Gamma)=r$. Then $(\xv',\yv')=f((\xv,\yv))$ defined as in \eqref{eq: supp_mult_prime} is an $r$-Goppa representation if and only if $f$ is an affine transformation.
\end{proposition}
\begin{proof}
	Inverting $f$, we obtain the relation
	\[
	\xv=\frac{a' \xv' +b'}{c' \xv' +d'},
	\]
	where 
	\[
	\begin{bmatrix} a' & b' \\ c' & d' \end{bmatrix} =\begin{bmatrix} a & b \\ c & d \end{bmatrix}^{-1}=\begin{bmatrix} d & -b \\ -c & a \end{bmatrix},
	\]
	thus leading to
	\[
	\xv=\frac{d \xv' -b}{-c \xv' +a}.
	\]
	Note that
	\[
	c \xv+d=\frac{cd\xv'-bc}{-c\xv'+a}+d=\frac{ad-bc}{-c \xv'+a}.
	\]
	The coordinates of the multiplier $\yv'$ can thus be formulated as the evaluation of a rational function in the coordinates of $\xv'$ as
	\begin{align*}
		\yv'=& \lambda (c\xv+d)^{r-1} \yv 	= \frac{\lambda (c\xv+d)^{r-1}}{\Gamma\left(\frac{d\xv'-b}{-c\xv'+a}\right)} = \frac{\lambda (c\xv+d)^{r-1}}{\sum_{i=0}^r \gamma_i \left(\frac{d\xv'-b}{-c\xv'+a}\right)^i} =\frac{\lambda (c\xv+d)^{r-1}(-c \xv'+a)^r}{\sum_{i=0}^r \gamma_i (d\xv'-b)^{i}(-c\xv'+a)^{r-i}}  \\
		=& \frac{\lambda (ad-bc)^{r-1}(-c \xv'+a)}{\sum_{i=0}^r \gamma_i (d\xv'-b)^{i}(-c\xv'+a)^{r-i}}. \\
	\end{align*}
	In other words, the reduced form of such rational function has in general a numerator of degree 1 and a denominator of degree $r$, \ie 
	\[
	\yv'= \frac{A\xv'+B}{\sum_{i=0}^r {\gamma_i}' (\xv')^i},
	\]
	with $A=-\lambda (ad-bc)^{r-1}c$ and $B=\lambda (ad-bc)^{r-1}a$. In particular, the sufficient and necessary condition for $(\xv',\yv')$ being an $r$-Goppa code representation is $A=0\iff c=0$. In this case 
	\[
	\xv'= \frac{a}{d} \xv + \frac{b}{d} \quad \text{and} \quad \yv'= \lambda d^{r-1} \yv.
	\]
\end{proof}

With the previous proposition at hand, we thus distinguish two cases:
\begin{itemize}
	\item The Goppa code is full-support, \ie $\{x_i \mid i \in \Iintv{1}{n}\}=\Fqm$. In this case $\{c x_i +d \mid i \in \Iintv{1}{n}\}=\{c z +d \mid z \in \Fqm\}$ does not contain 0 iff $c=0$ (and $d\ne 0$), \ie $f$ is an affine map. By the existence of a Goppa code representation and by Proposition~\ref{prop: f_affine}, it follows that $(\xv,\yv)$ was already a Goppa code representation and nothing must be done.
	\item The Goppa code is not full-support. By computing the interpolation polynomial on the pairs $(x_i, \frac{1}{y_i})$'s we can check whether $(\xv,\yv)$ is an $r$-Goppa code representation. If not, Proposition~\ref{prop: f_affine} implies that a necessary condition for $(\xv',\yv')$ being an $r$-Goppa code representation is that $c\ne 0$. Without loss of generality, we can thus assume $c=1$, take the linear factor $\lambda=1$ and consider the candidate maps:
	\[ (\xv,\yv) \mapsto (\xv',\yv')=(\frac{a\xv+b}{\xv+d}, (\xv+d)^{r-1} \yv).\] 
	A simple method to get one good map is to choose at random $(a,b,d)$ and compute the degree of the interpolation polynomial for the pairs $(x_i', \frac{1}{y_i'})$'s until this is equal to $r$. The degrees of freedom for the affine map show that the probability of finding a sought map is
	\[
	\frac{(q^m-1) q^m}{(q^m-1) q^{2m}}=q^{-m},
	\]
	that is typically linear in the inverse of the length $n$. By quotienting the maps with respect to the group of affine maps, the worst case complexity becomes $q^m$ times the cost of interpolation polynomial that can be upper bounded by $n^2$. From our experiments, we observed that few more than $r$ pairs must be interpolated to determine whether the polynomial has degree equal to or larger than $r$.  
\end{itemize}

\section{Description of structured matrices in the code of quadratic relationships}
Recall that $\cA$ is the canonical basis of of $\left(\Alt{r}{\xv}{\yv}^\perp\right)_{\Fqm}$. We will study the matrix code $\Cmat(\cA)$, whose elements have a special structure and are easier to treat. This will be instrumental to describe a set of rank-2 matrices in the matrix code of quadratic relationships for a binary Goppa code of degree $r=2$, whose dimension is 3 as a variety:

\begin{restatable}{proposition}{variety} \label{prop: variety}
	Let $\Cmat$ be the matrix code of quadratic relationships corresponding to the dual of an $[n,n-2m]$ binary Goppa code in the extension field with a square-free Goppa polynomial of degree $2$. Let $\cP_2^+(\Mm)$ be the corresponding Pfaffian ideal as defined in \eqref{eq: pfaffian_ideal+} and where $\Mm$ is the generic skew-symmetric matrix of size $2m\times 2m$. Then $\dim \Vm(\cP_2^+(\Mm))\ge 3$.
\end{restatable}
This proposition will be better explained and proved afterwards. 

We remark again that the following description is not limited to Goppa codes of degree $r=2$, but it also applies to alternant codes or higher degrees. Let us now consider the following block shape for an element of $\Am\in\Cmat(\cA)$

\[  
\Am= \begin{bmatrix}
	\Am_{0,0} & \trsp{\Am_{1,0}} & \cdots & \trsp{\Am_{m-1,0}} \\
	\Am_{1,0}& \Am_{1,1} & \cdots &  \trsp{\Am_{m-1,1}}\\
	\vdots & \vdots & \ddots & \vdots\\
	\Am_{m-1,0} & \Am_{m-1,1} & \cdots & \Am_{m-1,m-1} \\
\end{bmatrix},
\]
with $\Am_{l,l}\in \Sym(r, \Fqm)$, and give a description of the $r\times r$ blocks $\Am_{i,j}$'s below and intersecting the main diagonal, i.e. for $i\ge j$, for the matrix $\Am$ associated with a given algebraic relation. The blocks above the main diagonal can be easily obtained by transposition. In particular, we can identify identities that correspond to sparse and constant low-rank matrix codewords. We split such matrices into several types. 

\textbf{Type 1.} Let $a+b=2c$ and $0\le a<c<b\le r-1$. The quadratic relation
\[
\starp{(\xv^a \yv)^{q^l}}{(\xv^b \yv)^{q^l}}=\starp{(\xv^c \yv)^{q^l}}{(\xv^c \yv)^{q^l}}
\]
translates into $\starp{\av_a^{q^l}}{\av_b^{q^l}}-\starp{\av_c^{q^l}}{\av_c^{q^l}}=0$ with respect to the basis $\cA$ and thus it is associated with the matrix $\Am$ such that
	\[
	\Am_{l,l}=\begin{blockarray}{cccccc}
		& \textcolor{gray}{a} & \textcolor{gray}{c} & \textcolor{gray}{b} &  \\
		\begin{block}{[ccccc]c}
			&  &  &  &  &  \\
			\zerov &  &  & 1 &  & \textcolor{gray}{a} \\
			&  & -2 &  &  & \textcolor{gray}{c} \\
			& 1 &  &  & \zerov & \textcolor{gray}{b} \\
			&  &  &  & &  \\
		\end{block}
	\end{blockarray},\qquad \text{and }\Am_{i,j}=\zerov_{r\times r} \;\text{ for }\;  0\le j\le i\le m-1, (i,j)\ne (l,l).
	\] 
We remark that, in characteristic 2, $-2=0$, and thus $\Am_{l,l}$ has only two nonzero entries.
Therefore $\rank(\Am)=\rank(\Am_{l,l})=\begin{cases}
	3 & \text{if field characteristic $\ne 2$} \\ 2 & \text{if field characteristic $= 2$}\\ 
\end{cases}$.
We also observe that Type 1 matrices exist only for $r\ge 3$.

\textbf{Type 2.} Let $0\le a<c< d<b\le r-1$ and $a+b=c+d$. The quadratic relation
\[
\starp{(\xv^a \yv)^{q^l}}{(\xv^b \yv)^{q^l}}=\starp{(\xv^c \yv)^{q^l}}{(\xv^d \yv)^{q^l}}
\]
translates into $\starp{\av_a^{q^l}}{\av_b^{q^l}}-\starp{\av_c^{q^l}}{\av_d^{q^l}}=0$ with respect to the basis $\cA$ and thus it is associated with the matrix $\Am$ such that
\[
\Am_{l,l}=\begin{blockarray}{ccccccc}
	& \textcolor{gray}{a} & \textcolor{gray}{c} & \textcolor{gray}{d} & \textcolor{gray}{b}  & \\
	\begin{block}{[cccccc]c}
		&  &  &  &  & & \\
		\zerov & &  &  & 1 &  & \textcolor{gray}{a} \\
		& &  & -1 &  &  & \textcolor{gray}{c} \\
		& & -1  &  &  &  & \textcolor{gray}{d} \\
		& 1 &  & & & \zerov & \textcolor{gray}{b} \\
		& &  &  &  &  &  \\
	\end{block}
\end{blockarray},\qquad \text{and }\Am_{i,j}=\zerov_{r\times r}\; \text{ for }\;  0\le j\le i\le m-1, (i,j)\ne (l,l).
\] 
Therefore $\rank(\Am)=\rank(\Am_{l,l})=4$.
We observe that Type 1 matrices exist only for $r\ge 4$.
\begin{remark}
	For $a+b=c+d$ even, Type 2 matrices can be derived as linear combinations of two Type 1 matrices thanks to the equality
	\[
	\starp{\av_a^{2^l}}{\av_b^{2^l}}-\starp{\av_c^{2^l}}{\av_d^{2^l}}= (\starp{\av_a^{2^l}}{\av_b^{2^l}}-\starp{\av_{\frac{a+b}{2}}^{2^l}}{\av_{\frac{a+b}{2}}^{2^l}})-(\starp{\av_c^{2^l}}{\av_d^{2^l}}-\starp{\av_{\frac{c+d}{2}}^{2^l}}{\av_{\frac{c+d}{2}}^{2^l}})
	\]
	that holds under the conditions above.
	
	Moreover, we can obtain linear dependencies within the set of Type 2 matrices. Indeed let $a+b=c+d=e+f$ with $0\le a<c<e<f<d<b\le r-1$. Then the equality
	\[
	\starp{\av_a^{2^l}}{\av_b^{2^l}}-\starp{\av_c^{2^l}}{\av_d^{2^l}}= (\starp{\av_a^{2^l}}{\av_b^{2^l}}-\starp{\av_e^{2^l}}{\av_f^{2^l}})-(\starp{\av_c^{2^l}}{\av_d^{2^l}}-\starp{\av_e^{2^l}}{\av_f^{2^l}})
	\]
	induces a linear dependency on the matrices.
\end{remark}

\textbf{Type 3.} Let $q=2$ be the field size of a Goppa code with a square-free Goppa polynomial and $0\le a<b\le r-1$. In the quadratic relation
\[
\starp{(\xv^a \yv)^{2^l}}{(\xv^b \yv)^{2^l}}=\starp{(\xv^{a+b} \yv^2)^{2^{l-1}}}{(\xv^{a+b} \yv^2)^{2^{l-1}}},
\]
the term $\starp{(\xv^a \yv)^{2^l}}{(\xv^b \yv)^{2^l}}$ translates into $\starp{\av_a^{2^l}}{\av_b^{2^l}}$ with respect to the basis $\cA$, while $\starp{(\xv^{a+b} \yv^2)^{2^{l-1}}}{(\xv^{a+b} \yv^2)^{2^{l-1}}}$ is the square of a codeword in $\Goppa{\xv}{\Gamma}^\perp_{\Fqm}$. Indeed 
\[(\xv^{a+b} \yv^2)^{2^{l-1}}\in \GRS{2r}{\xv}{\yv^2}^{(2^{l-1})}\subseteq \Alt{2r}{\xv}{\yv^2}_{\Fqm}^\perp=\Alt{2r}{\xv}{1/\Gamma(\xv)^2}_{\Fqm}^\perp=\Goppa{\xv}{\Gamma}_{\Fqm}^\perp
\]
with the last equality following from Theorem~\ref{thm: binary_Goppa->Alt}. Therefore, as shown in \cite[Proposition~9]{CMT23}, this quadratic relation is associated with the matrix $\Am$ such that
\[
\Am_{l,l}=\begin{blockarray}{cccccc}
	& \textcolor{gray}{a} & & \textcolor{gray}{b} &  \\
	\begin{block}{[ccccc]c}
		&  &  &  &  &  \\
		\zerov &  &  & 1 &  & \textcolor{gray}{a} \\
		&  &  &  &  &  \\
		& 1 &  &  & \zerov & \textcolor{gray}{b} \\
		&  &  &  & &  \\
	\end{block}
\end{blockarray},\qquad \text{and }\Am_{i,j}=\zerov_{r\times r}\; \text{ for }\;  0\le j\le i\le m-1, (i,j)\ne (l,l).
\] 
Therefore $\rank(\Am)=\rank(\Am_{l,l})=2$.
\begin{remark}
	Type 1 and Type 3 matrices are the only ones described in \cite{CMT23} and are exactly those of rank $\le 3$. In the case of field characteristic 2, these are indeed the targets of the Pfaffian modeling. Note that, in the case of a square-free binary Goppa code, Type 1 matrices are included in Type 3 matrices and correspond to the choice $a+b$ being even. Moreover, Type 2 matrices can be obtained as linear combinations of two Type 3 matrices thanks to the equality
	\[
	\starp{\av_a^{2^l}}{\av_b^{2^l}}-\starp{\av_c^{2^l}}{\av_d^{2^l}}= (\starp{\av_a^{2^l}}{\av_b^{2^l}}-\starp{\av_{a+b}^{2^{l-1}}}{\av_{a+b}^{2^{l-1}}})-(\starp{\av_c^{2^l}}{\av_d^{2^l}}-\starp{\av_{c+d}^{2^{l-1}}}{\av_{c+d}^{2^{l-1}}})=0
	\]
	that holds whenever $a+b=c+d$.
\end{remark}

\textbf{Type 4.} Let $a<c$, $b>d$, $(u-l) \mod m \le e_{\AC}\eqdef \floor{\log_q(r+1)}$ and $aq^l+bq^u=cq^l+dq^u$. The quadratic relation
\[
\starp{(\xv^a \yv)^{q^l}}{(\xv^b \yv)^{q^u}}=\starp{(\xv^c \yv)^{q^l}}{(\xv^d \yv)^{q^u}}
\]
translates into $\starp{\av_a^{q^l}}{\av_b^{q^u}}-\starp{\av_c^{q^l}}{\av_d^{q^u}}=0$ with respect to the basis $\cA$ and thus it is associated with the matrix $\Am$ such that
\begin{itemize}
	\item if $u>l$:
	\[
	\Am_{u,l}=\begin{blockarray}{cccccc}
		& \textcolor{gray}{a} & & \textcolor{gray}{c} &  \\
		\begin{block}{[ccccc]c}
			&  &  &  &  &  \\
			\zerov &  &  & -1 &  & \textcolor{gray}{d} \\
			& 1 &  &  &  & \textcolor{gray}{b} \\
			&  &  &  & \zerov &  \\
			&  &  &  & &  \\
		\end{block}
	\end{blockarray},\qquad \text{and }\Am_{i,j}=\zerov_{r\times r}\; \text{ for }\;  0\le j\le i\le m-1, (i,j)\ne (u,l).
	\] 
	\item if $l>u$:
	\[
	\trsp{\Am_{u,l}}=\begin{blockarray}{cccccc}
		& \textcolor{gray}{a} & & \textcolor{gray}{c} &  \\
		\begin{block}{[ccccc]c}
			&  &  &  &  &  \\
			\zerov &  &  & -1 &  & \textcolor{gray}{d} \\
			& 1 &  &  &  & \textcolor{gray}{b} \\
			&  &  &  & \zerov &  \\
			&  &  &  & &  \\
		\end{block}
	\end{blockarray},\qquad \text{and }\Am_{i,j}=\zerov_{r\times r}\; \text{ for }\;  0\le j\le i\le m-1, (i,j)\ne (l,u).
	\]
\end{itemize}
Therefore $\rank(\Am)=\rank(\Am_{u,l})+\rank(\trsp{\Am_{u,l}})=4$. Differently from the previous types, these matrices are not block diagonal.
\begin{remark} \label{remark: lindep4} We can obtain linear dependencies within the set of Type 4 matrices. Indeed let $aq^l+bq^u=cq^l+dq^u=eq^l+fq^u$ with $a<c<e,f<d<b$. Then the equality
	\[
	\starp{\av_a^{q^l}}{\av_b^{q^u}}-\starp{\av_c^{q^l}}{\av_d^{q^u}}= (\starp{\av_a^{q^l}}{\av_b^{q^u}}-\starp{\av_e^{q^l}}{\av_f^{q^u}})-(\starp{\av_c^{q^l}}{\av_d^{q^u}}-\starp{\av_e^{q^l}}{\av_f^{q^u}})
	\]
	induces a linear dependency on the matrices.
\end{remark}
We also observe that Type 4 matrices exist only starting from $r\ge 3$.

\textbf{Type 5.} Let $q=2$ be the field size of a Goppa code with a square-free Goppa polynomial $\Gamma(z)=\sum_{a=0}^r \gamma_a z^a$. Fix $u,l\in \Iintv{0}{m-1}$ such that $1\le (u-l) \mod m \le v\eqdef\floor{\log_2 r}$ and let $s\eqdef 2^{(u-l-1) \mod m}w$, for some $w \in \Iintv{0}{2r-2}$. For all $a \in \Iintv{0}{r-1}$, we define 
\[B_a\eqdef \Iintv{\ceil{\frac{a+s-r+1}{2^{(u-l)\mod m}}}}{\floor{\frac{a+s}{2^{(u-l)\mod m}}}}\cap \Iintv{0}{r-1}.\] 
For all $a\in\Iintv{0}{r}$, let us also define $c_{(a,b)}=a+s-2^{(u-l)\mod m}b$ and $\gamma_{c_{(a,b)},b}\in \Fqm$, $b\in B_a$, such that $\sum_{b\in B_a} \gamma_{c_{(a,b)},b}=\gamma_a$. In the quadratic relation
\begin{align*}
	&\sum_{a=0}^r \left(\sum_{b \in B_a} \gamma_{c_{(a,b)},b}^{2^l} \starp{(\xv^{c_{(a,b)}} \yv)^{2^l}}{(\xv^b \yv)^{2^u}} \right)\\ = &\sum_{a=0}^r \left(\sum_{b \in B_a} \gamma_{c_{(a,b)},b} \xv^{a}\right)^{2^l} \xv^{ 2^l s} \yv^{2^u+2^l} \\ = & \xv^{ 2^l s}\yv^{2^u}\\= &\starp{(\xv^w \yv^2)^{2^{(u-2)\mod m}}}{(\xv^w \yv^2)^{2^{(u-2)\mod m}}},
\end{align*}
the term $\sum_{a=0}^r \left(\sum_{b \in B_a} \gamma_{c_{(a,b)},b}^{2^l} \starp{(\xv^{c_{(a,b)}} \yv)^{2^l}}{(\xv^b \yv)^{2^u}} \right)$ translates into\\ $\sum_{a=0}^r \left(\sum_{b \in B_a} \gamma_{c_{(a,b)},b}^{2^l} \starp{\av_{c_{(a,b)}} ^{2^l}}{\av_b^{2^u}} \right)$ with respect to the basis $\cA$, while\\ $\starp{(\xv^w \yv^2)^{2^{(u-2)\mod m}}}{(\xv^w \yv^2)^{2^{(u-2)\mod m}}}$ is the square of a codeword in $\Goppa{\xv}{\Gamma}^\perp_{\Fqm}$ because of Theorem~\ref{thm: binary_Goppa->Alt}. Therefore, analogously to Type 3 matrices, whenever $\forall a \in \Iintv{0}{r}, B_a \ne \emptyset$, this quadratic relation is associated with the matrix $\Am$ such that
\begin{itemize}
	\item if $u>l$:
	\[
	\Am_{u,l}=\begin{blockarray}{cccc}
		& \textcolor{gray}{i} &  \\
		\begin{block}{[ccc]c}
			\ddots & \vdots & \iddots & \\
			\cdots & \gamma_{i,b}^{2^l} & \cdots & \textcolor{gray}{b} \\
			\iddots & \vdots & \ddots & \\
		\end{block}
	\end{blockarray},\qquad \text{and }\Am_{i,j}=\zerov_{r\times r}\; \text{ for }\;  0\le j\le i\le m-1, (i,j)\ne (u,l).
	\] 
	\item if $l>u$:
	\[
	\trsp{\Am_{u,l}}=\begin{blockarray}{cccc}
		& \textcolor{gray}{i} &  \\
		\begin{block}{[ccc]c}
			\ddots & \vdots & \iddots & \\
			\cdots & \gamma_{i,b}^{2^l} & \cdots & \textcolor{gray}{b} \\
			\iddots & \vdots & \ddots & \\
		\end{block}
	\end{blockarray},\qquad \text{and }\Am_{i,j}=\zerov_{r\times r}\; \text{ for }\;  0\le j\le i\le m-1, (i,j)\ne (l,u).
	\]
	where $\gamma_{i,b}$ is set to 0 for any $i\not\in \{c_{(a,b)} \mid a \in \Iintv{0}{r}\}$.
\end{itemize}
This is the most articulated type of matrix, and in general it is not straightforward to calculate its rank. Since $B_a$ must be nonempty for every $a\in \Iintv{0}{r}$ and $\card{\Iintv{0}{r}}=r+1>r$, at least two rows of $\Am_{u,l}$ are nonzero and it can be verified that they cannot all be linearly dependent. This implies that $\rank(\Am)=2\rank(\Am_{u,l})\ge 4$. However, for any $r,u,l$ and for any admissible choice of $s$, we can choose the $\gamma_{c_{(a,b)},b}$'s in such a way that only two consecutive rows are nonzero. In particular, we can determine sparse matrices (with only $2(r+1)$ nonzero entries) of rank 4. We provide a couple of examples below.
\begin{example}
	Let $r=6$, $u=1$, $l=0$ and let us take $s=1$. The matrix $\Am$ such that
	\[
	\Am_{1,0}=\begin{bmatrix}
		0 & \gamma_0 & \gamma_1 & \gamma_2 & \gamma_3 & \gamma_4 \\
		0 & 0 & 0 & 0 & \gamma_5 & \gamma_6  \\
		0 & 0 & 0 & 0 & 0 & 0 \\ 		0 & 0 & 0 & 0 & 0 & 0 \\ 		0 & 0 & 0 & 0 & 0 & 0 \\ 		0 & 0 & 0 & 0 & 0 & 0 \\
	\end{bmatrix}
	\] 
	and $\Am_{i,j}=\zerov_{r\times r}$ for $0\le j\le i\le m-1, (i,j)\ne (u,l)$, belongs to $\Cmat(\cA)$ and has rank 4.
\end{example}
\begin{example}	
	Let $r=6$, $u=2$, $l=0$ and let us take $s=2\cdot 5=10$. The matrix $\Am$ such that
	\[
	\Am_{3,1}=\begin{bmatrix}
		0 & 0 & 0 & 0 & 0 & 0 \\ 		0 & 0 & 0 & 0 & 0 & 0 \\
		\gamma_0^2 & \gamma_1^2 & \gamma_2^2 & \gamma_3^2  & 0 & 0 \\
		\gamma_4^2 & \gamma_5^2 & \gamma_6^2 & 0 & 0 & 0 \\		0 & 0 & 0 & 0 & 0 & 0 \\ 		0 & 0 & 0 & 0 & 0 & 0 \\
	\end{bmatrix}
	\] 
	and $\Am_{i,j}=\zerov_{r\times r}$ for $0\le j\le i\le m-1, (i,j)\ne (u,l)$, belongs to $\Cmat(\cA)$ and has rank 4.
\end{example}

\subsection*{The Pfaffian variety with respect to a square-free Goppa polynomial of degree $r=2$}

Let $\F$ be a finite field of characteristic 2, $\CC\subset \F^n$ an $[n,s]$ linear code with basis $\Cmat(\cV)$, $\Mm=(m_{i,j})\in\Skew(s,\F)$ the generic skew-symmetric matrix (\ie all $\binom{s}{2}$ entries below the main diagonal are seen as independent variables) and $\mv$ the vector of length $\binom{s}{2}$ containing all the independent variables corresponding to the entries of $\Mm$. We recall here the definition of the generic Pfaffian ideal $\cP_2(\Mm)\in \F[\mv]$ for rank 2.
\begin{definition}[Pfaffian ideal for rank 2]
	The \textit{Pfaffian ideal} of rank 2 for $\Mm$ in characteristic 2 is
	\begin{equation} \label{eq: pfaffian_ideal}
		\cP_2(\Mm) \eqdef \langle m_{i,j}m_{k,l}+m_{i,k}m_{j,l}+m_{i,l}m_{j,k} \mid 1\le i<j<k<l\le s \rangle.
	\end{equation}
\end{definition}
The Hilbert series of the quotient $\F[\mv]/\cP_2(\Mm)$ is known \cite{GK04} and the dimension of the associated variety is $2s-3$.

In \cite{CMT23}, the Pfaffian ideal $\cP_2^+(\Mm)\in \F[\mv]$ for rank 2 with respect to the matrix code $\Cmat(\cV)$ has been defined. This can be written as
\begin{equation} \label{eq: pfaffian_ideal+}
\cP_2^+(\Mm)=\cP_2(\Mm)+\langle L_1(\mv),\dots,L_t(\mv)\rangle,
	\end{equation}
where the $t= \binom{s}{2}-\dim_{\F} \Cmat(\cV)$ linearly independent linear polynomials $L_i$'s in $\mv$ express the fact that $\Mm$ belongs to the matrix subspace $\Cmat(\cV)$. This ideal was introduced because, by evaluating the Hilbert function at a high enough degree, alternant and Goppa codes can be distinguished from random codes. In particular, if $\CC\subset \Fqm^n$ is a random $[n,rm]$ code and $n>3rm-3$, then with high-probability the variety $\Vm(\cP_2^+(\Mm))=\{\zerov\}$. On the other hand, if $\CC=\Goppa{\xv}{\Gamma}^\perp_{\F_2^m}$, where $\Goppa{\xv}{\Gamma}$ is an $[n,n-rm]$ binary Goppa code with a square-free Goppa polynomial defined from a field extension of degree $m$ and $r\ge 3$, then \cite[Proposition~16]{CMT23} shows that 
\begin{equation} \label{eq: dimVge2r-3}
	\dim \Vm(\cP_2^+(\Mm))\ge 2r-3.
\end{equation}
The reason why the result does not include the case $r=2$ is strongly connected to the distinguishability of the GRS code $\GRS{r}{\xv}{\yv}$. Indeed the smallest value for which two Schur's products of different pairs of codewords in $\GRS{r}{\xv}{\yv}$ coincide is 3: $\starp{(\xv^2 \yv)}{(\yv)}=\sqb{\xv \yv}$. Intuitively, this is reasonable: the vectors $\xv$ and $\yv$ are a compact representation of the GRS code, but any 2-dimensional code is determined by two linearly independent codewords. This fact is inherited when considering subfield subcodes of GRS codes: the alternant code $\Alt{2}{\xv}{\yv}$ is not distinguishable from random codes in general. However, a binary Goppa code of degree 2 with a square-free Goppa polynomial remains in principle distinguishable and $\Vm(\cP_2^+(\Mm))\supsetneq \{0\}$. This is obvious from the existence of rank-2 matrices in $\Cmat(\cA)$. Indeed, by fixing $l$ and choosing $a=0, b=1$ in the Type 3 construction, we obtain the matrix $\Am \in \Cmat(\cA)$ such that
\[
\Am_{l,l}=\begin{bmatrix} 0 & 1\\ 1& 0 \end{bmatrix},\qquad \qquad \Am_{i,j}=\zerov_{2\times 2}\quad  \text{ otherwise.}
\] 
This immediately implies that $\dim \Vm(\cP_2^+(\Mm))\ge 1$ getting back the lower bound in \eqref{eq: dimVge2r-3} for $r=2$ too. We recall and finally prove Proposition~\ref{prop: variety}, which shows that in this case the lower bound on the dimension of the variety is not tight. 
\variety*
\begin{proof}
	We consider a subspace of $\Cmat(\cA)$ spanned by some of the matrices presented in Section 2. Let us fix some $l\in \Iintv{0}{m-1}$ and take
	\begin{itemize}
		\item the Type 3 matrix obtained by taking $a=0,b=1$ and $l$;
		\item the Type 3 matrix obtained by taking $a=0,b=1$ and $l+1$;
		\item the Type 5 matrices obtained by taking $s=0,u=l+1,\gamma_{c_{(0,0)},0}=\gamma_0, \gamma_{c_{(1,0)},0}=\gamma_1$, $\gamma_{c_{(2,1)},1}=\gamma_2$ and $\gamma_{i,b}=0$ otherwise;
		\item the Type 5 matrices obtained by taking $s=1,u=l+1,\gamma_{c_{(0,0)},0}=\gamma_0, \gamma_{c_{(1,1)},1}=\gamma_1$, $\gamma_{c_{(2,1)},1}=\gamma_2$ and $\gamma_{i,b}=0$ otherwise.
	\end{itemize}
	We fix $l=0$ for simplicity, but the same argument works for any other value in $\Iintv{0}{m-1}$. Let $\lambda=(\lambda_1,\dots,\lambda_4)\in \Fqm^4$ be the vectors of coefficients in the linear combination of the 4 matrices above. We thus obtain the parametrized matrix
	\[
	\Am(\lambda)=\begin{bmatrix} 
		0 & \lambda_1 & \lambda_3\gamma_0 & \lambda_3\gamma_2+\lambda_4\gamma_1 & & &  \\
		\lambda_1 & 0 &  \lambda_3\gamma_1+\lambda_4\gamma_0 & \lambda_4\gamma_2 & & &\\
		\lambda_3\gamma_0 & \lambda_3\gamma_1+\lambda_4\gamma_0 & 0 & \lambda_2 & & \zerov_{4\times 2m-4} & \\
		\lambda_3\gamma_2+\lambda_4\gamma_1 & \lambda_4\gamma_2 & \lambda_2 & 0 & & & \\
		& & & & & & \\
		& \zerov_{2m-4\times 4} & & & & \zerov_{2m-4\times 2m-4} & \\
		& & & & & & \\
	\end{bmatrix}.
	\]
	The only possible nonzero $4\times 4$ minor is the top left one. This is a principal minor, thus the submatrix is skew-symmetric and its determinant is the square of a polynomial in its entries called \textit{pfaffian}, which we denote with $\pf$. A straightforward computation gives
	\begin{align*}
		&\det\left(\begin{bmatrix} 
			0 & \lambda_1 & \lambda_3\gamma_0 & \lambda_3\gamma_2+\lambda_4\gamma_1  \\
			\lambda_1 & 0 &  \lambda_3\gamma_1+\lambda_4\gamma_0 & \lambda_4\gamma_2\\
			\lambda_3\gamma_0 & \lambda_3\gamma_1+\lambda_4\gamma_0 & 0 & \lambda_2\\
			\lambda_3\gamma_2+\lambda_4\gamma_1 & \lambda_4\gamma_2 & \lambda_2 & 0
		\end{bmatrix}\right)\\=& \pf^2\left(\begin{bmatrix} 
			0 & \lambda_1 & \lambda_3\gamma_0 & \lambda_3\gamma_2+\lambda_4\gamma_1  \\
			\lambda_1 & 0 &  \lambda_3\gamma_1+\lambda_4\gamma_0 & \lambda_4\gamma_2\\
			\lambda_3\gamma_0 & \lambda_3\gamma_1+\lambda_4\gamma_0 & 0 & \lambda_2\\
			\lambda_3\gamma_2+\lambda_4\gamma_1 & \lambda_4\gamma_2 & \lambda_2 & 0
		\end{bmatrix}\right)\\=&(\lambda_1 \lambda_2+\lambda_3\lambda_4\gamma_0\gamma_2+\lambda_3^2\gamma_1\gamma_2+\lambda_3\lambda_4 \gamma_1^2+\lambda_3\lambda_4 \gamma_0\gamma_2+\lambda_4^2\gamma_0\gamma_1)^2\\
		=&(\lambda_1 \lambda_2+\lambda_3^2\gamma_1\gamma_2+\lambda_3\lambda_4 \gamma_1^2+\lambda_4^2\gamma_0\gamma_1)^2.\\
	\end{align*}
	We want to study for which $\lambda$ the pfaffian $\lambda_1 \lambda_2+\lambda_3^2\gamma_1\gamma_2+\lambda_3\lambda_4 \gamma_1^2+\lambda_4^2\gamma_0\gamma_1$ equals 0. For any free choice of $\lambda_3,\lambda_4$ two cases may occur:
	\begin{itemize}
		\item if $\lambda_3^2\gamma_1\gamma_2+\lambda_3\lambda_4 \gamma_1^2+\lambda_4^2\gamma_0\gamma_1=0$, then the equality is obtained by fixing one between $\lambda_1$ or $\lambda_2$ to 0 and the other can take any value over $\Fqm$;
		\item if $\lambda_3^2\gamma_1\gamma_2+\lambda_3\lambda_4 \gamma_1^2+\lambda_4^2\gamma_0\gamma_1\ne0$, then the equality is obtained by taking any $\lambda_1 \ne 0$ and $\lambda_2=\lambda_1^{-1}$.
	\end{itemize}
	Therefore, there are 3 degrees of freedom in the choice of $\lambda$ and since all the corresponding matrices belong to $\Cmat(\cA)$, this implies that $\dim \Vm(\cP_2^+(\Mm))\ge 3$.
\end{proof}

The previous proposition shows that we can fix 3 variables in the Pfaffian system and still expect to have non-trivial solutions, \ie rank-2 matrices in $\Cmat(\cB)$, with overwhelming probability in the specialized non-homogeneous system. 

\section{Retrieving the basis of a GRS code}
Analogously to the previous section, we focus here on the case of a binary Goppa code with a square-free Goppa polynomial of degree 2 and adapt the attack from  \cite[Section 6]{CMT23} to this setting.
We first recall the setting of the attack proposed in \cite{CMT23} for distinguishable parameters with respect to the notion of distinguishability given in \cite{FGOPT11, MT23}. In this framework, under the condition that $3\le r<q+1$ ($3\le r<q-1$ respectively) the matrix code $\Cmat(\cA)$ originated by a generic alternant code (generic Goppa code respectively) contains only block diagonal elements (with $m$ blocks of size $r\times r$). This very special shape allowed to recover a basis of the alternant/Goppa code, by sampling almost full-rank matrices in the code and gaining information about codewords lying in a single GRS code from the kernels of these matrices. 

However, the conditions are not satisfied here for several reasons:
\begin{itemize}
	\item the Goppa code is not distinguishable in the sense of \cite{FGOPT11, MT23};
	\item the Goppa polynomial degree $r=2<3$;
	\item for Goppa codes it was required that $r<q-1$, but here $r=q=2$.
\end{itemize}

We thus follow an alternative approach. Suppose we are able to sample rank-2 matrices in $\Cmat(\cB)$ by solving the Pfaffian system specialized into $3$ variables. In this section, we describe how to recover a basis of one of the $m$ codes $\GRS{2}{\xv^{q^j}}{\yv^{q^j}}$. The algorithm is therefore an adaptation of \cite[Section 6]{CMT23}, where kernels are computed from low-rank matrices instead of almost full-rank ones.

A set of rank-2 matrices in $\Cmat(\cA)$ has been exhibited in the previous section. Below the Gilbert-Varshamov distance computed with respect to the space of skew-symmetric matrices of size $rm$, all the solutions of the Pfaffian system are expected to have the structure described in the proof of Proposition~\ref{prop: variety}, or the analogous one for another fixed value of $l\in \Iintv{0}{m-1}$, with very high probability. We will exploit that any rank-2 matrix in the code $\Cmat(\cA)$ has this structure. We recall from \cite[Proposition 10]{CMT23} that the parameters above the Gilbert-Varshamov distance correspond to those for which
\[
n>3m-3.
\]

Therefore, the attack we will now present works for binary Goppa codes with a square-free Goppa polynomial of degree 2 and is based on the following standard assumption that we have extensively verified through experiments (on both TII challenges and randomly generated instances):
\begin{assumption} \label{assumption: rank2}
	Let $\Goppa{\xv}{\Gamma}$ be a binary $[n,n-2m]$ Goppa code with a square-free Goppa polynomial $\Gamma$ of degree $2$ and let $n>3m-3$. Let $\cA$ be the canonical basis of $\Goppa{\xv}{\Gamma}^\perp_{\Fqm}$ as given in \eqref{eq: basisA}. Then, for $n\to\infty$, a rank 2 matrix in $\Am\in\Cmat(\cA)$ is such that only 4 blocks are non-zero, in particular
	\[\exists l\in\Iintv{0}{m-1} \text{ s.t. } \forall (u,v)\notin\{(l,l),(l+1,l), (l+1,l+1)\}\mod m,\quad\Am_{u,v}=\zerov_{r\times r},
	\]
	\ie
	\begin{equation*} 
		\Am=\begin{bmatrix}
			\zerov_{r\times r} & & & & & \zerov_{r\times r}\\
			& \ddots & & & \iddots & \\
			& & \Am_{l-1,l-1} & \trsp{\Am_{l, l-1}} & & \\
			& & \Am_{l,l-1} & \Am_{l, l} & & \\
			& \iddots & & & \ddots & \\
			\zerov_{r\times r} & & & & & \zerov_{r\times r}\\
		\end{bmatrix}.
	\end{equation*}
	with probability $1-o(1)$.
\end{assumption}

In particular the shape of $\Am$ in Assumption~\ref{assumption: rank2} is the same (minus diagonal shift) as the one of $\Am(\lambda)$ in Proposition~\ref{prop: variety} for an admissible value of the parameter $\lambda$ that makes $\Am(\lambda)$ of rank 2.

Algorithm~\ref{alg: attack} provides a sketch of the attack.

\begin{algorithm}[t]
	\begin{algorithmic}[1]
		\State Choose a basis $\cB=(\bv_1, \dots, \bv_r, \bv_1^q, \dots, \bv_r^q, \dots, \bv_1^{q^{m-1}},\dots, \bv_r^{q^{m-1}})$ for $\ext{\Goppa{\xv}{\Gamma}^\perp}{ \Fqm}$.
		\Repeat 
		\State Find two matrices $\Bm_1,\Bm_2\in \Cmat(\cB)$ of rank $2$
		\State $K_1 \gets \ker(\Bm_1)$
		\State $K_2 \gets \ker(\Bm_2)$
		\State $\VC\gets K_1\cap K_2$
		\State $\VC\gets \VC+\VC^{(q)}\Sm$ \quad \Comment \text{$\Sm$ is defined as in \eqref{eq: matF}}
		\State $\mathbf{GRS}\gets\VC^\perp \cdot \Hm_{\cB}$ \quad \Comment \text{$\Hm_{\cB}$ is defined as in \eqref{eq: Hm}}
		\Until $\dim_{\Fqm} \mathbf{GRS}=2$ and $\mathbf{GRS}\subsetneq \Fq^n$
		\State Apply the Sidelnikov-Shestakov attack on $\mathbf{GRS}$
	\end{algorithmic} \caption{Sketch of the attack for binary square-free Goppa codes of order 2}\label{alg: attack}
\end{algorithm}

It succeeds if at some iteration of the repeat cycle the code $\mathbf{GRS}$ is one of the $m$ GRS codes of which $\Goppa{\xv}{\Gamma}$ is a subfield subcode. Indeed, in this case, the well-known Sidelnikov-Shestakov attack allows to retrieve a valid pair of support and multiplier for $\mathbf{GRS}$ that also defines an alternant code that coincides with $\Goppa{\xv}{\Gamma}$. 

Hence the key recovery of Algorithm~\ref{alg: attack} follows on the spot from the next proposition.
\begin{proposition} \label{prop: correctness}
	Consider an iteration of the \textit{repeat/until} loop of Algorithm~\ref{alg: attack}. For $n\to\infty$ and under Assumption~\ref{assumption: rank2}, the equation 
	\[(\VC+\VC^{(q)}\Sm)^\perp\cdot \Hm_{\cB}= \GRS{2}{\xv}{\yv}^{q^l}\]
	holds with probability $\frac{1}{m}(1-o(1))$ at least.
\end{proposition}
In the next subsection, we provide a proof of Proposition~\ref{prop: correctness} that is therefore a proof of correctness for Algorithm~\ref{alg: attack} under Assumption~\ref{assumption: rank2}.

\subsection{Explanation of Algorithm~\ref{alg: attack} under Assumption~\ref{assumption: rank2}}
Algorithm~\ref{alg: attack} starts by computing a basis $\cB$ of $\ext{\Goppa{\xv}{\Gamma}^\perp}{ \Fqm}$ with a special shape:
\begin{equation} \label{eq: AltBasisFrobenius}
	\mathcal{B} = (\bv_1, \dots, \bv_r, \bv_1^q, \dots, \bv_r^q, \dots, \bv_1^{q^{m-1}},\dots, \bv_r^{q^{m-1}}).
\end{equation}
An efficient procedure to obtain a basis like this has already been explained in \cite{CMT23}. 
Given a basis $\cV=(\vv_1,\dots,\vv_k)$ of an $[n,k]$ code (for instance $\cA$ or $\cB$) we denote by $\Hm_{\cV}$ the $k\times n$ matrix whose rows are the elements of $\cV$ in the same order:
\begin{equation} \label{eq: Hm}
\Hm_{\cV}\eqdef \begin{pmatrix} \vv_1 \\ \vdots \\ \vv_k \end{pmatrix}.
\end{equation}

Let us define the right $r$-cyclic shift matrix  $\Sm\in \mathbf{GL}_{mr}(\Fqm)$ as
\begin{equation} \label{eq: matF}
	\Sm \eqdef
	\begin{pmatrix}
		& \mat{I}_r &  & & \\
		&& \mat{I}_r & &\zerov \\
		&&\zerov & \ddots & \\
		&& & & \mat{I}_r\\
		\mat{I}_r &  &  && \\
	\end{pmatrix}.
\end{equation}
Note that $\Sm^{-1}=\trsp{\Sm}$ is the left $r$-cyclic shift matrix. 
We first recall, without proving it, a preliminary result from \cite{CMT23} that also comes in handy here. 
\begin{lemma} \label{prop: stable_frobenius}
	Whenever a basis $\cB$ has the form given in \eqref{eq: AltBasisFrobenius}, 
	$\Cmat (\mathcal B)$ is stable
	by the operation
	\[
	\Mm \longmapsto \trsp{\Sm} \Mm^{(q)} \Sm.
	\]
\end{lemma}

The next results characterize the structure of the kernel of a rank 2 matrix $\Am\in\Cmat(\cA)$.
\begin{lemma} \label{prop: ker_sum_GRS}
	Let $\cA,\cB$ be the two bases introduced before and $\Pm$ the change of basis, \ie $\Hm_{\cB}=\Pm \Hm_{\cA}$. Let $\Bm\in\Cmat(\cB)$ be of rank 2 and $n>3m-3$. Then $\exists\, l \in \Iintv{0}{m-1}$ such that
	\begin{equation*} \label{eq:Vperp}
		\ker(\Bm) \trsp{(\Pm^{-1})} \Pm^{-1} \Hm_{\cB} \supset \sum_{j\in\Iintv{0}{m-1}\setminus\{(l-1 \mod m),l\}} \GRS{2}{\xv}{\yv}^{q^j}
	\end{equation*}
	with probability $1-o(1)$.
\end{lemma}
\begin{proof}
	For better readability, we will assume in the following that $l\in\Iintv{1}{m-2}$, but the same arguments work for $l=0,m-1$ as well. Let $\trsp{\Pm}\Bm\Pm\in \Cmat(\cA)$. From Assumption~\ref{assumption: rank2}, with overwhelming probability, $\trsp{\Pm}\Bm\Pm$ has only 4 non-zero blocks, in particular
	\begin{equation} \label{eq: 2x2block}
		\trsp{\Pm}\Bm\Pm=\begin{bmatrix}
			\zerov_{r\times r} & & & & & \zerov_{r\times r}\\
			& \ddots & & & \iddots & \\
			& & \Am_{l-1,l-1} & \trsp{\Am_{l, l-1}} & & \\
			& & \Am_{l,l-1} & \Am_{l, l} & & \\
			& \iddots & & & \ddots & \\
			\zerov_{r\times r} & & & & & \zerov_{r\times r}\\
		\end{bmatrix}.
	\end{equation}
	Let $\cv\in \sum_{j\in\Iintv{0}{m-1}\setminus\{(l-1 \mod m),l\}} \GRS{2}{\xv}{\yv}^{q^j}$. Then $\cv= \dv\Hm_{\cA}$, where
	\[
	\dv=(\zerov_r,\dots,\zerov_r,\dv_l,\dv_{l+1},\zerov_r,\dots,\zerov_r).
	\]
	This implies that
	\[
	\dv\in\ker(\trsp{\Pm}\Bm\Pm),
	\]
	which is equivalent to
	\[
	\dv\trsp{\Pm}\in\ker(\Bm),
	\]
	since $\Pm$ is invertible. Hence, we obtain
	\[
	\cv = \dv\Hm_{\cA}=\dv (\trsp{\Pm} \trsp{(\Pm^{-1})}) \Pm^{-1} \Hm_{\cB} \in 	\ker(\Bm) \trsp{(\Pm^{-1})} \Pm^{-1} \Hm_{\cB}.
	\]
\end{proof}
\begin{lemma} \label{lemma: Vq}
	Let $\Bm\in \Cmat(\cB)$, $\VC=\ker(\Bm)$ and $\Sm$ be defined as in \eqref{eq: matF}. Then $\VC^{(q)}\Sm=\ker(\trsp{\Sm}\Bm^{(q)}\Sm).$
\end{lemma}
\begin{proof}
	This readily follows from the fact that $\Sm$ is invertible with inverse $\trsp{\Sm}$. Indeed, $\forall \vv\in \VC$,
	\begin{align*}
		&0=\vv \Bm\\
		\iff & 0=(\vv \Bm)^q \Sm= \vv^q \Bm^{(q)}\Sm= (\vv^q \Sm)\cdot \trsp{\Sm}\Bm^{(q)}\Sm.
	\end{align*}
\end{proof}

Let $\Bm_1$ and $\Bm_2$ be the two matrices sampled at line 3 in Algorithm~\ref{alg: attack} and consider the case where $\trsp{\Pm}\Bm_1\Pm$ and $\trsp{\Pm}\Bm_2\Pm$ are as in Equation~\eqref{eq: 2x2block} for the same indices $l-1 \mod m$ and $l$. This will happen with probability $1/m$. Then, since $\dim_{\Fqm} \ker(\Bm_1)=\dim_{\Fqm} \ker(\Bm_2)=rm-2$ and $\Bm_1$ and $\Bm_2$ are independent, we expect that $\dim_{\Fqm} \ker(\Bm_1)\cap \ker(\Bm_2)=rm-4$ with high probability. If this is the case, the fact that 
\[\dim_{\Fqm} \sum_{j\in\Iintv{0}{m-1}\setminus\{(l-1 \mod m),l\}} \GRS{2}{\xv}{\yv}^{q^j}=rm-4\]
in conjunction with Proposition~\ref{prop: ker_sum_GRS} implies that
\begin{equation} \label{eq: V=sumGRS}
	\VC  \trsp{(\Pm^{-1})} \Pm^{-1} \Hm_{\cB} = \sum_{j\in\Iintv{0}{m-1}\setminus\{(l-1 \mod m),l\}} \GRS{2}{\xv}{\yv}^{q^j},
\end{equation}
where 
\[
\VC \eqdef (\ker(\Bm_1) \cap \ker(\Bm_2)).
\]
We are finally ready to prove Proposition~\ref{prop: correctness} under Assumption~\ref{assumption: rank2}, thus showing how the space $\VC$ unveils a basis for a single GRS code.

\begin{proof}[Proof of Proposition~\ref{prop: correctness}.] From Assumption~\ref{assumption: rank2} and the argument shown above about matrices $\Bm_1$ and $\Bm_2$, we have that Equation~\eqref{eq: V=sumGRS} holds with probability at least $\frac{1}{m}(1-o(1))$. In this case, by Lemmata~\ref{prop: ker_sum_GRS} and \ref{lemma: Vq}, we deduce that
	\[
	\VC^{(q)}\Sm  \trsp{(\Pm^{-1})} \Pm^{-1} \Hm_{\cB}=\sum_{j\in\Iintv{0}{m-1}\setminus\{l,(l+1 \mod m)\}} \GRS{2}{\xv}{\yv}^{q^j}.
	\]
	Under the usual assumption that all the $m$ GRS codes have trivial intersection, we obtain
	\[
	(\VC+\VC^{(q)}\Sm) \trsp{(\Pm^{-1})} \Pm^{-1} \Hm_{\cB}=\sum_{j\in\Iintv{0}{m-1}\setminus\{l\}} \GRS{2}{\xv}{\yv}^{q^j}.
	\]
	
	Let us now pick $\vv^\perp \in (\VC+\VC^{(q)}\Sm)^\perp$. For any $\vv \in \VC+\VC^{(q)}\Sm$, we can write
	\[
	0= \langle \vv, \vv^\perp \rangle 
	=\langle \vv \Im_{rm}, \vv^\perp \rangle 
	=\langle \vv (\trsp{\Pm})^{-1} \trsp{\Pm}, \vv^\perp \rangle 
	=  \langle \vv\trsp{(\Pm^{-1})} , \vv^\perp \Pm \rangle .
	\]
	Therefore $\vv^\perp \Pm$ is zero outside the $j$-th block. Hence
	\[
	(\VC+\VC^{(q)}\Sm)^\perp \Hm_{\cB}= ((\VC+\VC^{(q)}\Sm)^\perp \Pm) \Hm_{\cA} \subseteq  \GRS{2}{\xv}{\yv}^{q^l},
	\]
	and since $\dim_{\Fqm}((\VC+\VC^{(q)}\Sm)^\perp)=rm-\dim_{\Fqm}(\VC+\VC^{(q)}\Sm)=2$, we get
	\[
	(\VC+\VC^{(q)}\Sm)^\perp \Hm_{\cB}=  \GRS{2}{\xv}{\yv}^{q^l}.
	\]
\end{proof}
At this stage, it is enough to apply the Sidelnikov-Shestakov attack \cite{SS92} on $\GRS{2}{\xv}{\yv}^{q^l}=(\VC+\VC^{(q)}\Sm)^\perp \Hm_{\cB}$. The support-multiplier pair output by this procedure is also a valid support-multiplier pair for the Goppa code $\Goppa{\xv}{\Gamma}=\Alt{2}{\xv}{\yv}$.

\begin{remark}
	It is easy to see that a slightly refined version of Algorithm~\ref{alg: attack} guarantees its termination, under Assumption~\ref{assumption: rank2}. Even if the sampled matrices $\Bm_1$ and $\Bm_2$ do not lead to the same shape of Equation~\eqref{eq: 2x2block}, there must be $l\in \Iintv{0}{m-1}$ such that $\Bm_1$ and $(\trsp{\Sm})^l\Bm_2^{(q^l)}\Sm^l$ do. Therefore, at most $m$ iterations are needed in order to get the GRS code.
\end{remark}

We conclude the section by giving the parameters of various TII challenges broken by this attack in Table~\ref{table: TII}. We run experiments in MAGMA on an Intel Core i7-1355U Processor. The variance of the attack timing computed on several resolutions is pretty high, due to the randomicity in the variable specialization, but on average all instances such that $n > 3rm-3$ take less than 10 seconds. When instead $n \le 3rm-3$, Assumption~\ref{assumption: rank2} does not hold anymore and rank-2 matrices other than those described in Proposition~\ref{prop: variety} are expected to belong to the matrix code. However, the attack is still supposed to work if the good matrices are sampled. Therefore, if $3m-3-n$ is a small natural number, the attack may still be practical, despite not having polynomial-time complexity anymore. Table~\ref{table: TII} provides a couple of such examples. 
\begin{table} 
\begin{center}
	\begin{tabular}{|c| c| c| c|c| c| c|} 
		\hline
		$q$ & $r$ & claimed bit complexity $\lambda$ & $m$ & $n$ & $n > 3rm-3$ ? & average time attack \\ 
		\hline\hline
		\multirow{10}{*}{2} &  \multirow{10}{*}{2} & 22 & 5 & 32 & yes & $<$3s\\  

		& & 39 & 5 & 28 & yes & $<$3s\\
		& & 41 & 5 & 27 & no (equal) & $\sim$10s\\
		& & 43 & 5 & 26 & no & $<$1min\\
		& & 44 & 6 & 61 & yes & $<$10s\\	
		& & 48 & 6 & 60 & yes & $<$10s\\
		& & 58 & 6 & 57 & yes & $<$10s\\
		& & 63 & 6 & 55 & yes & $<$10s\\	
		& & 65 & 6 & 54 & yes & $<$10s\\	
		& & 68 & 6 & 53 & yes & $<$10s \\ 	\hline
	\end{tabular}
\end{center}\caption{TII challenges with Goppa polynomial degree 2}  \label{table: TII}
\end{table}

\section{Conclusions}

In this paper, we analyzed in detail the matrix code of quadratic relationships, introduced in \cite{CMT23}, originated by a Goppa code. We described and categorized structured matrices with low rank, relating them to polynomial identities.

We extended the approach used in \cite{CMT23} to break instances of binary Goppa codes of degree 2, thus solving two TII challenges in a matter of a few seconds. To this aim, we first studied the variety associated with the Pfaffian ideal, proving that its dimension is at least 3. Then, from solutions of the Pfaffian system obtained by specializing 3 variables, we devised an efficient algorithm to reconstruct a GRS code that builds upon the strategy of \cite{CMT23}. This demonstrates the effectiveness of the Pfaffian modeling not only for distinguishing purposes but also for mounting key-recovery attacks.

Finally, we introduced the notion of Goppa code representation for Goppa codes and provided a procedure to get one of them from a generic pair of support and multiplier.






\section*{Acknowledgments}
The authors would like to thank Jean-Pierre Tillich and the anonymous reviewers for their careful reading of the manuscript and for their valuable comments and suggestions, which helped to reorganize the material and significantly improve its quality.









\newcommand{\etalchar}[1]{$^{#1}$}

\end{document}